\documentclass[twocolumn,twoside]{IEEEtran}

\usepackage[dvipsnames]{xcolor}
\usepackage{pdflscape}

\usepackage{colortbl}

\usepackage[latin1]{inputenc}
\usepackage{amssymb}
\usepackage{amsmath}
\usepackage[]{graphicx}
\usepackage{dsfont}
\usepackage{multirow}

\newcommand{\ie}{\emph{i.e.}}
\newcommand{\eg}{\emph{e.g.}}

\newtheorem{theorem}{Theorem}
\newtheorem{lemma}[theorem]{Lemma}

\newtheorem{corollary}[theorem]{Corollary}

\usepackage{accentbx}

\newcommand{\ps}[2]{\langle{#1},{#2}\rangle}
\newcommand{\psH}[2]{\langle{#1},{#2}\rangle_{\H}}

\newcommand{\normH}[1]{\|{#1}\|_\H}

\newcommand{\normHBig}[1]{\Big\|{#1}\Big\|_\H}

\graphicspath{{matlab/}}

\def\R{\ensuremath{\mathds{R}}}

\def\X{\ensuremath{\mathds{X}}} 
\def\H{\ensuremath{\mathds{H}}} 

\newcommand{\cb}[1]{{\ifmmode {\boldsymbol{#1}}\else ${\boldsymbol{#1}}$\fi}}
\newcommand{\cp}[1]{\ifmmode {\mathcal{#1}}\else ${\mathcal{#1}}$\fi}

\newcommand{\bx}{\cb{x}}
\newcommand{\bxi}{\cb{\xi}}
\newcommand{\bK}{{\cb{K}}}

\newcommand{\bkappa}{\cb{\kappa}}

\begin{document}


\title{Entropy of Overcomplete Kernel Dictionaries}

\author{Paul~Honeine,~\IEEEmembership{Member~IEEE}
\thanks{P.~Honeine is with the Institut Charles Delaunay (CNRS), Universit\'e de technologie de Troyes, 10000, Troyes, France. Phone: +33(0)325715625; Fax: +33(0)325715699; E-mail: paul.honeine@utt.fr
}}

\markboth{}
{Honeine: Entropy of Overcomplete Kernel Dictionaries}



\sloppy

\maketitle

\begin{abstract}
In signal analysis and synthesis, linear approximation theory considers a linear decomposition of any given signal in a set of atoms, collected into a so-called dictionary. Relevant sparse representations are obtained by relaxing the orthogonality condition of the atoms, yielding overcomplete dictionaries with an extended number of atoms. More generally than the linear decomposition, overcomplete kernel dictionaries provide an elegant nonlinear extension by defining the atoms through a mapping kernel function (\eg, the gaussian kernel). Models based on such kernel dictionaries are used in neural networks, gaussian processes and online learning with kernels. 

The quality of an overcomplete dictionary is evaluated with a diversity measure the distance, the approximation, the coherence and the Babel measures. In this paper, we develop a framework to examine overcomplete kernel dictionaries with the entropy from information theory. Indeed, a higher value of the entropy is associated to a further uniform spread of the atoms over the space. For each of the aforementioned diversity measures, we derive lower bounds on the entropy. Several definitions of the entropy are examined, with an extensive analysis in both the input space and the mapped feature space.
\end{abstract}

\begin{keywords} 
Generalized R\'enyi entropy, Shannon entropy, sparse approximation, dictionary learning, kernel-based methods, Gram matrix, machine learning, pattern recognition.
\end{keywords}

%
\IEEEpeerreviewmaketitle

%
%

\section{Introduction}

\PARstart{S}{parsity} in representation has gained increasing popularity in signal and image processing, for pattern recognition, denoising and compression \cite{Elad2010book}. A sparse representation of a given signal consists in decomposing it on a set of elementary signals, called atoms and collected in a so-called dictionary. In the linear formalism, the signal is written as a linear combination of the dictionary atoms. This decomposition is unique when the latter defines a basis, and in particular with orthogonal dictionaries such as with the Fourier basis. Since the 1960's, there has been much interest in this direction with the use of predefined dictionaries, based on some analytical form, such as with the wavelets 
\cite{Mal98}. Predefined dictionaries have been widely investigated in the literature for years, owing to the mathematical simplicity of such structured dictionaries when dealing with orthogonality (as well as bi-orthogonality). When dealing with sparsity, analytical dictionaries perform poorly in general, due to their rigide structure imposed by the orthogonality. 

Within the last 15 years, a new class of dictionaries has emerged with dictionaries learned from data, thus with the ability to adapt to the signal under scrutiny. While the Karhunen-Lo\`eve transform --- also called principal component analysis in advanced statistics \cite{PCA} --- falls in this class, the relaxation of the orthogonality condition delivers an increased flexibility with overcomplete dictionaries, \ie, when the number of atoms (largely) exceeds the signal dimension. Several methods have been proposed to construct oversomplete dictionaries by solving a highly non-convex optimization problem, such as the method of optimal directions \cite{Engan99}, its singular-value-decomposition (SVD) counterpart \cite{KSVD}, and the ``convexification'' method \cite{Mairal09}.

Overcomplete dictionaries are more versatile to provide relevant representations, owing to an increased diversity. Several measures have been proposed to ``quantify'' the diversity of a given dictionary. The simplest measure of diversity is certainly the cardinality of the dictionary, \ie, the number of atoms. While this measure is too simplistic, several diversity measures have been proposed by examining relations between atoms, either in a pairwise fashion or in a more thorough way. The most used measure to characterize a dictionary is the coherence, which is the largest pairwise correlation between its atoms \cite{Tropp06Just_Relax}. By using the largest cumulative correlation between an atom and all the other atoms of the dictionary, this yields the more exhaustive Babel measure \cite{Tro04}. Over the last twenty years or so, the coherence and its variants (such as the Babel measure) have been used for the matching pursuit algorithm \cite{Mallat93matchingpursuit} and the basis pursuit with arbitrary dictionaries \cite{Don03}, with theoretical results on the approximation quality studied in \cite{Gil03a,Tro04}; see also the extensive literature on compressed sensing \cite{Elad2010book}.

Beyond the literature on linear approximation, several diversity measures for overcomplete dictionary analysis have been investigated separately in the literature, within different frameworks. This is the case of the distance measure, which corresponds to the smallest pairwise distance between all atoms, as often considered in neural networks. Indeed, in resource-allocating networks for function interpolation, the network of gaussian units is assigned a new unit if this unit is {\em distant enough} to any other unit already in the network \cite{Platt91,Vukovic2013}. It turns out that these units operate as atoms in the approximation model, with the corresponding dictionary having a small distance measure. While the distance measure of a given dictionary relies only on its nearest pair of atoms, a more thorough measure is the approximation measure, which corresponds to the least error of approximating any atom of the dictionary with a linear combination of its other atoms. This measure of diversity has been investigated in machine learning with gaussian processes \cite{Csato02}, online learning with kernels for nonlinear adaptive filtering \cite{Pokharel2009}, and more recently kernel principal component analysis \cite{12.tpami}.

In order to provide a framework that encloses all the aforementioned methods, we consider the \emph{reproducing kernel Hilbert space} formalism. This allows to generalize the well-known linear model used in sparse approximation to a nonlinear one, where each atom is substituted by a nonlinear one given with a kernel function. This yields the so-called kernel dictionaries, where each atom lives in a feature space, the latter being defined with some nonlinear transformation of the input space. While the linear kernel yields the conventional linear model, as given in the literature of linear sparse approximation, the use of nonlinear kernels such as the gaussian kernel, allows to include in our study neural networks with ressource-allocating networks, nonlinear adaptive filtering with kernels and gaussian processes. 

All the aforementioned diversity measures allow to quantify the heterogeneity within the dictionary under scrutiny. In this paper, we derive connections between these measures and the entropy in information theory (which is also related to the definition of entropy in other fields, such as thermodynamics and statistical mechanics) \cite{Cover06}. Indeed, the entropy measures the disorder or randomness within a given system. By considering the generalized R\'enyi entropy, which englobes the definitions given by Shannon, Hartley, as well as the quadratic formulation, we show that any overcomplete kernel dictionary with a given diversity measure has a lower-bounded entropy. These results on the high values of the entropy illustrate that the atoms are favorably spread uniformly over the space. We provide a comprehensive analysis, for any kernel type and any entropy definition, within the R\'enyi entropy framework as well as the more recent nonadditive entropy proposed by Tsallis \cite{Tsallis09,Tsallis_entropy}. Finally, we provide an entropy analysis in the feature space by deriving lower bounds depending on the diversity measures. As a consequence, we connect the diversity measures between both input and feature spaces.

The remained of this paper is organized as follows. Next section introduces the sparse approximation problem, in its conventional linear model as well as its nonlinear extension with the kernel formalism. Section~\ref{sec:measures} presents the most used diversity measures for quantifying overcomplete dictionaries, while Section~\ref{sec:fundamental} provides a preliminary exploration with results that will be used throughout this paper. Section~\ref{sec:entropy} is the core of this work, where we define the entropy and examine it in the input space, while Section~\ref{sec:entropy_H} extends this analysis to the feature space. Section~\ref{sec:final_remarks} concludes this paper.

\subsection*{Related work}

In \cite{Gir02}, Girolami considered the estimation of the quadratic entropy with a set of samples, by using the Parzen estimator based on a normalized kernel function. This formulation was investigated in regularization networks, and in particular \emph{least-squares support vector machines} (LS-SVM), in order to reduce the computational complexity by pruning samples that do not contribute sufficiently to the entropy \cite{LSSVM}. More recently, an online learning scheme was proposed in \cite{Jung06} for LS-SVM by using the approximation measure as a sparsification criterion. In our paper, we derive the missing connections between this criterion and the entropy maximization.

Richard, Bermudez and Honeine considered in \cite{Ric09.tsp} the analysis of the quadratic entropy of a kernel dictionary in terms of its coherence. We provide in our paper a framework to analyse overcomplete dictionaries with a more extensive examination, in both input and feature spaces, and generalizing to other entropy definitions and all types of kernels. The conducted analysis examines several diversity measures, including, but not limited to, the coherence measure.


\section{A primer on overcomplete (kernel) approximation}

In this section, we introduce the sparse approximation problem, in its conventional linear model as well as the kernel-based formulation. We conclude this section with an outline of the issues addressed in this paper.

\subsection{A primer on sparse approximation}

Consider a Banach space $\X$ of $\R^d$, denoted input space. The approximation theory studies the representation of a given signal $\bx$ of $\X$ with a dictionary of atoms (\ie, set of elementary signals), $\bx_1, \bx_2, \ldots, \bx_n \in \X$, and estimating their fractions in the signal under scrutiny. In linear approximation, the decomposition takes the form:
\begin{equation}\label{eq:linear}
	\bx \approx \sum_{i=1}^n \alpha_i \, \bx_i.
\end{equation}
This representation is unique when the atoms form a basis, by approximating the signal with its projection onto the span of the atoms, namely $\alpha_i = \bx_i^\top \bx$. Examples that involve orthonormal bases include the Fourier transform and the discrete cosine transform, as well as the data-dependent 
Karhunen-Lo\`eve transform (\ie, the PCA).

Beyond these orthogonal bases, the relaxation of the orthogonality provides more flexibility with the use of overcomplete dictionaries, which allows to investigate different constraints more properly, such as the sparsity of the representation. In this case, the coefficients $\alpha_i$ in \eqref{eq:linear} are obtained by promoting the sparsity of the representation. This optimization problem is often called sparse coding, assuming that the dictionary is known. In view of the vector $[\alpha_1 ~~ \alpha_2 ~~ \cdots ~~ \alpha_n]^\top$, sparsity can be promoted by minimizing its $\ell_0$ pseudo-norm, which counts the number of non-zero entries, or its $\ell_1$ norm, which is the closest convex norm to the $\ell_0$ pseudo-norm 
\cite{Has09}. 

Since the seminal work \cite{dict_learn96} where Olshausen and Field considered learning the atoms from a set of available data, data-driven dictionaries have been widely investigated. A large class of approaches have been proposed to solve iteratively the optimization problem by alternating between the dictionary learning (\ie, estimating the atoms $\bx_i$) and the sparse coding (\ie, estimating the coefficients $\alpha_i$). The former problem is essentially tackled 
with the maximum likelihood principle of the data or the maximum {\em a posteriori} probability of the dictionary. The latter corresponds to the sparse coding problem. The best known methods for solving the optimization problem\footnote{In practice, one has several signals $\bx$ in order to construct the dictionary, resulting in a Frobenius norm minimization.} 
\begin{equation}\label{eq:linear_pb}
	\arg\mathop{\min_{\alpha_i, \bx_i}} _{i=1\cdots n} 
	\Big\| \bx - \sum_{i=1}^n \alpha_i \, \bx_i \Big\|^2,
\end{equation}
subject to some sparsity promoting constraint, are the method of optimal directions \cite{Engan99} and the K-SVD algorithm \cite{KSVD}, where the dictionary is determined respectively with the Moore-Penrose pseudo-inverse and the SVD scheme. For more details, see \cite{Elad2010book} and references therein. It is worth noting that the sparsity constraint yields a difficult optimization problem, even when the model is linear in both coefficients and atoms.

\subsection{Kernel-based approximation}

Nonlinear models provide a more challenging issue. The formalism of \emph{reproducing kernel Hilbert spaces} (RKHS) provides an elegant and efficient framework to tackle nonlinearities. To this end, the signals $\bx_1, \bx_2, \ldots, \bx_n$ are mapped with a nonlinear function into some feature space $\H$, as follows:
\begin{eqnarray*}
	\X &\mapsto& \H \\
	\bx_i &\to& \kappa(\bx_i,\cdot)
\end{eqnarray*}
Here, $\kappa\!:\X \times \X\rightarrow\R$ is a positive definite kernel and the feature space $\H$ is the so-called reproducing kernel Hilbert space. Let $\psH{\cdot}{\cdot}$ and $\normH{\cdot}$ denote respectively the inner product and norm in the induced space $\H$. This space has some interesting properties, such as the reproducing property which states that any function $\psi(\cdot)$ of $\H$ can be evaluated at any $\bx_i$ of $\X$ using $\psi(\bx_i) = \psH{\psi(\cdot)}{\kappa(\bx_i,\cdot)}$. Moreover, we have the kernel trick, that is $\psH{\kappa(\bx_i,\cdot)}{\kappa(\bx_j,\cdot)} = \kappa(\bx_i,\bx_j)$ for any $\bx_i,\bx_j \in \X$. In particular, $\normH{\kappa(\cdot,\bx_i)}^2 = \kappa(\bx_i,\bx_i)$.

Kernels can be roughly divided in two categories, projective kernels as functions of the data inner product (\ie, $\ps{\bx_i}{\bx_j}$), and radial kernels as functions of their distance (\ie, $\|\bx_i - \bx_j\|$). The most used kernels and there expressions are given in \tablename~\ref{tab:kernels}. From these kernels, only the gaussian and the radial-based exponential kernels are unit-norm, that is $\normH{\kappa(\bx,\cdot)}=1$ for any $\bx \in \X$. In this paper, we do not restrict ourselves to a particular kernel. We denote 
\begin{equation*}
	{~r^2 = \inf_{\bx \in \X} \kappa(\bx,\bx)~}
		\qquad \text{and} \qquad 
	 {~R^2 = \sup_{\bx \in \X} \kappa(\bx,\bx),~}
\end{equation*}
where $\kappa(\bx,\bx) = \normH{\kappa(\bx,\cdot)}^2$. For unit-norm kernels, we get $R=r=1$.

While the linear kernel yields the conventional model given in \eqref{eq:linear}, nonlinear kernels such as the gaussian kernel provide the models investigated in RBF neural networks, gaussian processes \cite{gpml} and kernel-based machine learning \cite{Shawetaylor_Cristianini}, including the celebrated support vector machines \cite{Vap95}. For a set of data $\bx_1, \bx_2, \ldots, \bx_n \in \X$ and a given kernel $\kappa(\cdot,\cdot)$, the induced RKHS $\H$ is defined such as any element $\psi(\cdot)$ of $\H$ takes the form
\begin{equation}\label{eq:repr}
	\psi(\cdot) = \sum_{i=1}^{n} \alpha_i \, \kappa(\bx_i,\cdot).
\end{equation}
When dealing with an approximation problem in the same spirit of \eqref{eq:linear}-\eqref{eq:linear_pb}, the element $\psi(\cdot)$ is approximated by $\kappa(\bx,\cdot)$. Compared to the linear case given in \eqref{eq:linear}, it is easy to see that the above model is still linear in the coefficients $\alpha_i$, as well as the ``atoms'' $\kappa(\bx_i,\cdot)$, while it is nonlinear with respect to $\bx_i$. Indeed, the resulting optimization problem consists in minimizing the residual in the RKHS, with
\begin{equation*}
	\arg\mathop{\min_{\alpha_i, \bx_i}} _{i=1\cdots n} 
	 \normHBig{\kappa(\bx,\cdot) - \sum_{i=1}^n \alpha_i \, \kappa(\bx_i,\cdot) }^2.
\end{equation*}
On the one hand, the estimation of the coefficients is similar to the one given in the linear case with \eqref{eq:linear_pb}; the classical (linear) sparse coders can be investigated for this purpose. On the other hand, the dictionary determination is more difficult, since the model is nonlinear in the $\bx_i$; thus, conventional techniques such as the K-SVD algorithm can no longer be used. It turns out that the estimation of the elements in the input space is a tough optimization problem, known in the literature as the pre-image problem \cite{11.spm}. More recently, the authors of \cite{13.spl.dictionary,14.sp.dictionary} adjusted the elements $\bx_i$ in the input space for nonlinear adaptive filtering with kernels. In another context, the authors of \cite{14.mlsp.nmf,14.tpami.knmf} estimated these elements for the kernel non-negative matrix factorization.

\begin{table}
\caption{The most used kernels with their expressions, including tunable parameters $p,\sigma>0$ and $c\geq 0$. These kernels are grouped in two categories: projective kernels as functions of $\ps{\bx_i}{\bx_j}$, and radial kernels as functions of $\|\bx_i - \bx_j\|$.}
\begin{center}
\renewcommand{\arraystretch}{1.4} 
\begin{tabular}{clc}
& Kernel & $\kappa(\bx_i,\bx_j)$  \\\hline\hline
\multirow{3}{*}{\rotatebox{90}{projective~}} & Linear & $\ps{\bx_i}{\bx_j}$ \\
& Polynomial & $\left(\ps{\bx_i}{\bx_j} + c\right)^p$\\
& Exponential & $\exp\left(\ps{\bx_i}{\bx_j}\right)$ \\
\hline
\multirow{3}{*}{\rotatebox{90}{radial~}} & Inverse multiquadratic & $\left(\|\bx_i - \bx_j\|^2 + \sigma \right)^{-p}$ \\
& Exponential & $\exp\left(\frac{-1}{\sigma}\|\bx_i - \bx_j\|\right)$ \\
& Gaussian & $\exp\left(\frac{-1}{2\sigma^2} \|\bx_i - \bx_j\|^2\right)$ \\
\hline\hline
\end{tabular}
\end{center}
\label{tab:kernels}
\end{table}

\subsection{Addressed issues}

In either analysis or synthesis of overcomplete (kernel) dictionaries, 
with the grow in the number of atoms, an increase in the heterogeneity of the atoms is needed. Such diversification requires that the atoms are not too ``close'' to each other. Depending on the definition of closeness, several diversity measures have been proposed in the literature. This is the case when the closeness is given in terms of the metric, as given with the distance measure for a pairwise measure between atoms, or the approximation measure for a more thorough measure. This is also the case when the collinearity of the atoms is considered, such as with the coherence and the Babel measures. These diversity measures are described in detail in Section~\ref{sec:measures}, within the formalism for a kernel dictionary $\{\kappa(\bx_1,\cdot), \kappa(\bx_2,\cdot), \ldots, \kappa(\bx_n,\cdot)\}$.

In this paper, we connect these diversity measures to the entropy from information theory \cite{Cover06}. Indeed, from the viewpoint of information theory, the set $\{\bx_1, \bx_2, \ldots, \bx_n\}$ can be viewed as a finite source alphabet. A fundamental measure of information is the entropy, which quantifies the disorder or randomness of a given system or set. It is also associated to the number of bits needed, in average, to store or communicate the set under investigation
. A detailed definition of the entropy is given in Section~\ref{sec:entropy}, with connections between the entropy of the set $\{\bx_1, \bx_2, \ldots, \bx_n\}$ and the aforementioned diversity measures of the associated kernel dictionary $\{\kappa(\bx_1,\cdot), \kappa(\bx_2,\cdot), \ldots, \kappa(\bx_n,\cdot)\}$. Several entropy definitions are also investigated, including the generalized R\'enyi entropy and the Tsallis entropy. Finally, Section~\ref{sec:entropy_H} extends this analysis to the RKHS, by studying the entropy of set of atoms $\{\kappa(\bx_1,\cdot), \kappa(\bx_2,\cdot), \ldots, \kappa(\bx_n,\cdot)\}$.


\section{Diversity measures}\label{sec:measures}

In this section, we present measures that quantify the diversity of a given dictionary $\{\kappa(\bx_1,\cdot), \kappa(\bx_2,\cdot), \ldots, \kappa(\bx_n,\cdot)\}$. 
Each diversity measure is associated to a sparsification criterion for online learning, in order to construct dictionaries with large diversity measures. 

\subsection{Cardinality}

The cardinality of the dictionary, namely the number $n$ of atoms, is the simplest measure. However, such measure does not take into account that some atoms can be close to each others, \eg, duplicata.


\subsection{Distance measure}\label{sec:distance}

A simple measure to characterize a dictionary is the smallest distance between all pairs of its atoms, namely
\begin{equation*}
    \mathop{\min_{i,j=1\cdots n}}_{i \neq j} \normH{\kappa(\bx_{i},\cdot) -  \kappa(\bx_{j},\cdot)},
\end{equation*}
where
\begin{equation}\label{eq:dist_equality}
	\normH{\kappa(\bx_{i},\cdot) - \kappa(\bx_{j},\cdot)}^2
	 = \kappa(\bx_{i},\bx_{i}) -  2\, \kappa(\bx_{i},\bx_{j})
	  + \, \kappa(\bx_{j},\bx_{j}).
\end{equation}
In the following, we consider a tighter measure by using the distance between any two atoms, up to a scaling factor, which is a tighter measure since we have
\begin{equation}\label{eq:dist_equiv}
    \normH{\kappa(\bx_{i},\cdot) - \kappa(\bx_{j},\cdot)}
    \geq \min_{\xi} \normH{\kappa(\bx_{i},\cdot) - \xi \kappa(\bx_{j},\cdot)}.
\end{equation}
A dictionary is said to be $\delta$-distant when
\begin{equation*}
   \delta = \mathop{\min_{i,j=1\cdots n}}_{i \neq j} \min_{\xi} \normH{\kappa(\bx_{i},\cdot) - \xi \, \kappa(\bx_{j},\cdot)}.
\end{equation*}
Since the above distance is equivalent to the residual error of approximating any atom by its projection onto another atom, the optimal scaling factor $\xi$ takes the value $\kappa(\bx_{i},\bx_{j}) / \kappa(\bx_{j},\bx_{j})$, yielding 
\begin{equation*}
    \delta^2 = \mathop{\min_{i,j=1\cdots n}}_{i \neq j}
    \Bigg( \kappa(\bx_{i},\bx_{i}) - \frac{\kappa(\bx_{i},\bx_{j})^2}{\kappa(\bx_{j},\bx_{j})} \Bigg).
\end{equation*}
When dealing with unit-norm atoms, this expression boils down to $\delta^2 = 1 - \kappa(\bx_{i},\bx_{j})^2$.

A sparsification criterion for online learning is studied in ressource-allocating networks \cite{Platt91,Huang05ageneralized} with the ``novelty criterion'', by imposing a lower bound on the distance measure of the dictionary. Thus, any candidate atom is included in the dictionary if the distance measure of the latter does not fall below a given threshold that controls the level of sparseness.

\subsection{Approximation measure}\label{sec:approx}

While the distance measure relies only on the nearest two atoms, the approximation measure provides a more exhaustive analysis 
by quantifying the capacity of approximating any atom with a linear combination of the other atoms of the dictionary. A dictionary is said to be $\delta$-approximate if the following is satisfied:
\begin{equation}\label{eq:approx}
    \delta = \min_{i=1\cdots n} \min_{\xi_1\cdots \xi_n}\normHBig{\kappa(\bx_{i},\cdot) - \mathop{\sum_{j=1}^n}_{j \neq i} \xi_j\,\kappa(\bx_{j},\cdot)}.
\end{equation}
This expression corresponds to the residual error of projecting any atom onto the subspace spanned by the others atoms. By nullifying the derivative of the above cost function with respect to each coefficient $\xi_j$, we get the optimal vector of coefficients
\begin{equation}\label{eq:approx.proj}
    \bxi = \bK_{\!_{\setminus\!\{\!i\!\}}\!}^{-1} \bkappa_{\!_{\setminus\!\{\!i\!\}}\!}(\bx_{i}),
\end{equation}
Here, $\bK_{\!_{\setminus\!\{\!i\!\}}\!}$ and $\bkappa_{\!_{\setminus\!\{\!i\!\}}\!}(\bx_{i})$ are obtained by removing the entries associated to $\bx_{i}$ from $\bK$ and $\bkappa(\bx_{i})$, respectively, where $\bK$ is the Gram matrix of entries $\kappa(\bx_i,\bx_j)$ and $\bkappa(\cdot)$ is the column vector of entries $\kappa(\bx_j,\cdot)$, for $i,j=1,\ldots, n$. By plugging the above expression in \eqref{eq:approx}, we obtain:
\begin{equation}\label{eq:approx1}
	\delta^2 = 
	\min_{i=1\cdots n}
	\kappa(\bx_{i},\bx_{i}) - \bkappa_{\!_{\setminus\!\{\!i\!\}}\!}(\bx_{i})^\top \bK_{\!_{\setminus\!\{\!i\!\}}\!}^{-1} \bkappa_{\!_{\setminus\!\{\!i\!\}}\!}(\bx_{i}).
\end{equation}

The sparsification criterion associated to the approximation measure is studied in \cite{Bau01,Csato01} and more recently in \cite{Eng04} for system identification and \cite{12.tpami} for kernel principal component analysis. This criterion constructs dictionaries with a high approximation measure, thus including any candidate atom in the dictionary if it cannot be well approximated by a linear combination of atoms already in the dictionary, for a given approximation threshold.

\subsection{Coherence measure} \label{sec:coherence}

In the literature of sparse linear approximation, the coherence is a fundamental quantity to characterize dictionaries. It corresponds to the largest correlation between atoms of a given dictionary, or mutually between atoms of two dictionaries. While initially introduced for linear matching pursuit in \cite{Mallat93matchingpursuit}, it has been studied for the union of two bases \cite{Donoho01uncertaintyprinciples}, for basis pursuit with arbitrary dictionaries \cite{Don03}, for the analysis of the approximation quality \cite{Gil03a,Tro04}. While most work consider the use of a linear measure, we explore in the following the coherence of a kernel dictionary, as initially studied in \cite{Hon07.isit}.

For a given dictionary, the coherence is defined by the largest correlation between all pairs of atoms, namely
\begin{equation*}
\mathop{\max_{i,j=1\cdots n}}_{i \neq j} 
\frac{|\psH{\kappa(\bx_{i},\cdot)}{\kappa(\bx_{j},\cdot)}|}
        {\normH{\kappa(\bx_{i},\cdot)} \normH{\kappa(\bx_{j},\cdot)}}.
\end{equation*}
It is easy to see that this definition can be written, for a so-called $\gamma$-coherent dictionary, as follows:
\begin{equation}\label{eq:coher}
	\gamma = \mathop{\max_{i,j=1\cdots n}}_{i \neq j} 
	\frac{|{\kappa(\bx_{i},\bx_{j})}|} {\sqrt{\kappa(\bx_{i},\bx_{i}) \, \kappa(\bx_{j},\bx_{j})}},
\end{equation}
For unit-norm atoms, we get 
	$\displaystyle\mathop{\max_{i,j=1\cdots n}}_{i \neq j} 
	|\kappa(\bx_{i},\bx_{j})|$.

The coherence criterion for sparsification constructs a ``low-coherent'' dictionary, thus enforcing an upper bound on the cosine angle between each pair of atoms \cite{Ric09.tsp}. In this case, any candidate atom is included in the dictionary if the coherence of the latter does not exceed a given threshold.This threshold controls the level of sparseness of the dictionary, where a null value yields an orthogonal basis.

\subsection{Babel measure}\label{sec:Babel}

While the coherence relies only on the most correlated atoms in the dictionary, a more thorough measure is the Babel measure which considers the largest cumulative correlation between an atom and all the other atoms of the dictionary. The Babel measure can be defined in two ways. The first one is by connecting it to the coherence measure, with a definition related to the cumulative coherence, namely
\begin{equation}\label{eq:babelD_normalized}
\max_{i=1\cdots n} \mathop{\sum_{j=1}^n}_{j \neq i}
\frac{|{\kappa(\bx_{i},\bx_{j})}|} {\sqrt{\kappa(\bx_{i},\bx_{i}) \, \kappa(\bx_{j},\bx_{j})}}.
\end{equation}
The second (and most conventional) way to define the Babel measure is by investigating an analogy with the norm operator \cite{Gil03b,Tro04}. Indeed, while the coherence is the $\infty$-norm of the Gram matrix when dealing with unit-norm atoms, the Babel measure explores the $\ell_1$ matrix-norm, where $\|\bK\|_1 = \max_i \sum_j |\kappa(\bx_{i}, \bx_{j})|$. As a consequence, a dictionary is said to be $\gamma$-Babel when
\begin{equation}\label{eq:babelD}
      \gamma
      = \max_{i=1\cdots n} \mathop{\sum_{j=1}^n}_{j \neq i} |\kappa(\bx_{i}, \bx_{j})|.
\end{equation}
Connecting this definition with \eqref{eq:babelD_normalized} --- for not necessary unit-norm atoms --- is straightforward, since the latter can be box-bounded for any $\gamma$-Babel dictionary defined by \eqref{eq:babelD}, with
 \begin{equation*}
\frac{\gamma}{R^2} \leq \max_{i=1\cdots n} \mathop{\sum_{j=1}^n}_{j \neq i}
\frac{|{\kappa(\bx_{i},\bx_{j})}|} {\sqrt{\kappa(\bx_{i},\bx_{i}) \, \kappa(\bx_{j},\bx_{j})}}
\leq \frac{\gamma}{r^2}.
 \end{equation*}
For this reason and for the sake of simplicity, we consider the definition \eqref{eq:babelD} in this paper. 

The sparsification criterion associated to the Babel measure constructs dictionaries with a low cumulative coherence \cite{Fan2014}. To this end, any candidate atom $\kappa(\bx_t,\cdot)$ is included in the dictionary if (and only if)
\begin{equation}\label{eq:babel.crit}
      \sum_{j=1}^n |\kappa(\bx_t, \bx_{j})|
\end{equation}
does not exceed a given positive threshold.
 
\section{Some fundamental results}\label{sec:fundamental}

Before proceeding throughout this paper with a rigorous analysis of any overcomplete dictionary in terms of its diversity measure, we provide in the following some results that are essential to our study. These results provide an attempt to bridge the gap between the different diversity measures.

\subsection{Coherence versus Babel measure}

The following theorems connect the coherence of a dictionary to its Babel measure by quantifying the Babel measure of a $\gamma$-coherent dictionary, and vice-versa. The following theorem has been known for a while in the case of unit-norm atoms.
\begin{theorem}\label{th:babel.coher}
	A $\gamma$-coherent dictionary has a Babel measure that does not exceed $(n-1)\gamma R^2$.
\end{theorem}
\begin{proof}
Following the definition \eqref{eq:babelD}, the Babel of a $\gamma$-coherent dictionary is upper-bounded as follows:
\begin{align*}
& \max_{i=1\cdots n} \mathop{\sum_{j=1}^n}_{j \neq i} |\kappa(\bx_{i}, \bx_{j})|\\
 	&\qquad\qquad \leq (n-1) \mathop{\max_{i,j=1\cdots n}}_{i \neq j} | \kappa(\bx_{i},\bx_{j}) |\\
	&\qquad\qquad \leq (n-1) \gamma \mathop{\max_{i,j=1\cdots n}}_{i \neq j}\sqrt{\kappa(\bx_{i},\bx_{i}) \, \kappa(\bx_{j},\bx_{j})}\\
	&\qquad\qquad \leq (n-1) \gamma R^2.
\end{align*}
\end{proof}

Furthermore, it is also easy to provide an upper bound on the coherence of a dictionary with a given Babel measure, as given in the following theorem.
\begin{theorem}\label{th:coher.babel}
	A $\gamma$-Babel dictionary has a coherence that does not exceed $\gamma/r^2$.
\end{theorem}
\begin{proof}
The proof follows from the relation
\begin{align*}
\mathop{\max_{i,j=1\cdots n}}_{i \neq j} 
	\frac{|{\kappa(\bx_{i},\bx_{j})}|} {\sqrt{\kappa(\bx_{i},\bx_{i}) \, \kappa(\bx_{j},\bx_{j})}}
&\leq \mathop{\max_{i,j=1\cdots n}}_{i \neq j} 
	\frac{|{\kappa(\bx_{i},\bx_{j})}|}{r^2},
\end{align*}
and the inequality between matrix norms: 
	$\|\cdot\|_{\max} \leq \|\cdot\|_\infty.$
\end{proof}

\subsection{Analysis of a $\delta$-approximate dictionary}

The following theorem is fundamental in the analysis of a dictionary resulting from the approximation criterion.

\begin{theorem}\label{th:approx}
	A $\delta$-approximate dictionary has a Babel measure that does not exceed 
	$R^2 - \delta^2,$	
	and a coherence measure that does not exceed
	$$\frac{R^2 - \delta^2}{r^2}.$$
\end{theorem}

\begin{proof}
For a $\delta$-approximate dictionary, we have from \eqref{eq:approx.proj}: $\bK_{\!_{\setminus\!\{\!i\!\}}\!} \bxi = \bkappa_{\!_{\setminus\!\{\!i\!\}}\!}(\bx_{i})$, for any $i=1,2, \ldots, n$. By plugging this relation in \eqref{eq:approx1}, we obtain
\begin{equation*}
	\min_{\bxi} \kappa(\bx_{i},\bx_{i}) - \bkappa_{\!_{\setminus\!\{\!i\!\}}\!}(\bx_{i})^\top \bxi \geq \delta^2.
\end{equation*}
By considering the special case of the vector $\bxi$ with $\xi_j=\mathrm{sign}(\kappa(\bx_{i},\bx_{j}))$, for any $j=1,2, \ldots, n$ and $j\neq i$, we get
\begin{equation*}
	 \mathop{\sum_{j=1}^n}_{j \neq i} |\kappa(\bx_{i},\bx_{j})|
	\leq \kappa(\bx_{i},\bx_{i}) - \delta^2,
\end{equation*}
for all $i=1,2, \ldots, n$. As a consequence, 
\begin{equation*}
	 \max_{i=1\cdots n} \mathop{\sum_{j=1}^n}_{j \neq i} |\kappa(\bx_{i},\bx_{j})|
	\leq \max_{i=1\cdots n} \kappa(\bx_{i},\bx_{i}) - \delta^2
	\leq R^2 - \delta^2.
\end{equation*}
This concludes the proof for the Babel measure, since it is the left-hand-side in the above expression, while the upper bound on the coherence measure is obtained from the aforementioned connection between the coherence and the Babel measures as given in Theorem~\ref{th:coher.babel}.
\end{proof}

%
%
%


\section{Entropy analysis in the input space}\label{sec:entropy}

The entropy measures the disorder or randomness within a given system. The R\'enyi entropy provides a generalization of well-known entropy definitions, such as Shannon and Harley entropies as well as the quadratic entropy (see \tablename~\ref{tab:entropy}). It is defined for a given order $\alpha$ by 
\begin{equation}\label{eq:entropy_int}
	H_\alpha = \frac{1}{1-\alpha} \log \int_\X \big(P(\bx)\big)^\alpha \, d\bx,
\end{equation}
for the probability distribution $P$ that governs all elements $\bx$ of $\X$. When dealing with discrete random variables as in source coding, 
this definition is restricted to the set $\{\bx_{1}, \bx_{2}, \ldots, \bx_{n}\}$ drawn from the probability distribution $P$, yielding the expression
\begin{equation}\label{eq:entropy_sum}
	H_\alpha = \frac{1}{1-\alpha} \log \sum_{j=1}^n \big(P(\bx_{j})\big)^\alpha.
\end{equation}
Large values of the entropy correspond to a more uniform spread of the data\footnote{It is well-known for the Shannon entropy (\ie, where $\alpha \to 1$) that the uniform distribution yields the largest entropy. This property seems to extend to the case of any non-zero order, including $\alpha \to \infty$ where we get the min-entropy. See \tablename~\ref{tab:entropy} for the expressions of well-known entropies.}. Since this probability distribution is unknown in practice, it is often approximated with a Parzen window estimator (also called kernel density estimator). The estimator takes the form 
\begin{equation}\label{eq:P_estim}
	\widehat{P}(\bx)=\frac{1}{n}\sum_{j=1}^n w(\|\bx - \bx_{j}\|),
\end{equation}
for a given window function $w$ centered at each $\bx_{j}$. For more details, see for instance \cite{KECA}.

In the following, we provide lower bounds on the entropy of an overcomplete dictionary, in terms of its diversity measure. To this end, we initially restrict ourselves to the case of the quadratic entropy (\ie, $\alpha = 2$), first with the gaussian kernel then with any type of kernel, before generalizing these results to any order $\alpha$ of the R\'enyi entropy as well as the Tsallis entropy.


\subsection{The quadratic entropy with the gaussian kernel}\label{sec:entropy_gaussian}

Before generalizing to any window function in Section~\ref{sec:entropy_quadratic} and any order in Section~\ref{sec:entropy_Renyi}, we restrict ourselves first to the case of the gaussian window function with the quadratic entropy. The quadratic entropy is defined by $H_2 = -\log \sum_{j=1}^n \big(P(\bx_{j})\big)^2$. Considering the normalized gaussian window 
\begin{equation*}
     w(\|\bx - \bx_{j}\|)=\frac{1}
        {(\sqrt{\pi}\sigma)^d}\,\exp\left(-\|\bx-\bx_{j}\|^2/\sigma^2 \right),
\end{equation*}
for some bandwidth parameter $\sigma$, the Parzen estimator becomes
\begin{equation*}
     \widehat{P}(\bx)=\frac{1}{n}\sum_{j=1}^n\frac{1}
        {(\sqrt{\pi}\sigma)^d}\,\exp\left(-\|\bx-\bx_{j}\|^2/\sigma^2 \right).
\end{equation*}
Since the convolution of two gaussian distributions leads to another gaussian distribution, then $H_2 \approx-\log \sum_{j=1}^n \big(\widehat{P}(\bx_{j})\big)^2$ becomes
\begin{align}\label{eq:entropy_gaussian}
 \nonumber
     H_2 
     & \approx	 -\log \Bigg(\frac{1}{n^2}\sum_{i,j=1}^n\frac{\kappa(\bx_{i},\bx_{j})}{(2\pi\sigma^2)^{d/2}} \Bigg)
\\	 &= \frac{d}{2} \log\!\left(2\pi\sigma^2\right)
	 -\log \Bigg(\frac{1}{n^2}\sum_{i,j=1}^n \kappa(\bx_{i},\bx_{j}) \Bigg)
	 ,
\end{align}
where $\kappa(\bx_{i},\bx_{j}) = \exp\left(\frac{-1}{2\sigma^2}\|\bx_{i}-\bx_{j}\|^2\right)$ is the gaussian kernel. This expression shows that the sum of the entries in the Gram matrix describes the diversity of the dictionary elements, a result corroborated in \cite{Gir02} and more recently in \cite{KECA}. This property was investigated in \cite{LSSVM} for pruning the LS-SVM, by removing samples with the smallest entries in the Gram matrix.

Each diversity measure studied in Section~\ref{sec:measures} yields a lower bound on the entropy of the dictionary under scrutiny. To shown this, we consider first the Babel measure with a $\gamma$-Babel dictionary. Following the Babel definition in \eqref{eq:babelD}, the entropy given in \eqref{eq:entropy_gaussian} is lower-bounded as follows:
\begin{equation*}
     H_2 \geq \frac{d}{2} \log\!\left(2\pi\sigma^2\right)
	 + \log n - \log(1+ \gamma),
\end{equation*}
where we have used the following upper bound on the summation: 
\begin{align*}
	\sum_{i,j=1}^n \kappa(\bx_{i},\bx_{j}) 
	&= \sum_{i=1}^n \kappa(\bx_{i},\bx_{i})
		+ \sum_{i=1}^n\mathop{\sum_{j=1}^n}_{j \neq i} \kappa(\bx_{i}, \bx_{j})
\\	&\leq n ( 1+ \gamma).
\end{align*}
This result provides the core of the proof. Indeed, Theorem~\ref{th:coher.babel} shows that this result holds also for a $\gamma$-coherent dictionary. Furthermore, we can improve this bound for the coherence measure, since $\sum_{i=1}^n \sum_{j=1,j \neq i}^n \kappa(\bx_{i}, \bx_{j}) \leq n(n-1) \gamma$, thus yielding the following lower bound on the entropy
\begin{equation*}
     H_2 \geq \frac{d}{2} \log\!\left(2\pi\sigma^2\right)
	 + \log n - \log \big( 1+ (n-1)\gamma \big).
\end{equation*}
This result is also shared with a $\delta$-distant dictionary, by substituting $\gamma$ with $\sqrt{1- \delta^2}$, since the distance is  equivalent to the coherence when dealing with normalized kernels. Finally, Theorem~\ref{th:approx} establishes the connection with a $\delta$-approximate dictionary, where the above upper bound becomes
\begin{equation*}
     H_2 \geq \frac{d}{2} \log\!\left(2\pi\sigma^2\right)
	 + \log n - \log \left(2 - \delta^2 \right).
\end{equation*}

All these results provide lower bounds on the entropy, with the following observations. These bounds increase with the number of elements in the dictionary, \ie, $n$, which is obvious as the diversity grows. They 
decrease when the coherence and the Babel measures increase, while they increase when the distance and the approximation measures increase. These results provide quantitative details that confront the fact that, when using a sparsification criterion for online learning, low values of the coherence and Babel thresholds provide less ``correlated'' atoms and thus more diversity within the dictionary, as opposed to 
high values of the distance and approximation thresholds.

\begin{table}
\caption{The most known entropies as special cases of the generalized R\'enyi entropy.}
\begin{center}
\renewcommand{\arraystretch}{1.5} 
\begin{tabular}{@{}l@{\qquad}c@{\qquad}c@{}}
Entropy & order $\alpha$ & $H_\alpha$  \\\hline\hline
Harley entropy & $\alpha = 0$ & $\log n$ \\
Shannon entropy & $\alpha \to 1$ & $\displaystyle - \sum_{j=1}^n P(\bx_{j}) \log P(\bx_{j})$ \\
Quadratic entropy & $\alpha = 2$ & $\displaystyle - \log \sum_{j=1}^n \big(P(\bx_{j})\big)^2$ \\
Min-entropy & $\alpha \to \infty $ & $\displaystyle\min_{j=1\cdots n} -\log P(\bx_{j})$\\
\hline\hline
\end{tabular}
\end{center}
\label{tab:entropy}
\end{table}

\subsection{The quadratic entropy with any kernel}\label{sec:entropy_quadratic}

The results presented so far can be extended to any kernel, even non-unit-norm kernels. To see this, we define the Parzen estimator in a RKHS, by writing the integral $\int_\X \widehat{P}(\bx)^2 \, d\bx$ as the quadratic norm $\|\widehat{P}\|_\cp{H}^2$ of 
$$\widehat{P}(\cdot)=\frac{1}{n}\sum_{i=1}^n\kappa(\bx_{i},\cdot),$$ 
where the norm is given in the subspace spanned by the kernel functions $\kappa(\bx_{1},\cdot), \kappa(\bx_{2},\cdot), \ldots, \kappa(\bx_{n},\cdot)$. Therefore, we have 
\begin{equation*}
    H_2 \approx-\log\|\widehat{P}\|_\cp{H}^2
        = - \log \left(\frac{1}{n^2}\sum_{i,j=1}^n\kappa(\bx_{i},\bx_{j})\right).
\end{equation*}
By following the same steps as in Section~\ref{sec:entropy_gaussian}, 
 we can derive the following lower bounds on the quadratic entropy: 
\begin{itemize}
\item $\log n -\log \big(R^2 +  (n-1) R \sqrt{R^2 - \delta^2} \big)$ for a $\delta$-distant dictionary.
\item $\log n - \log \big(2R^2 - \delta^2\big)$ for a $\delta$-approximate dictionary.
\item $\log n - \log \big(R^2+ (n-1)\gamma R^2\big)$ for a $\gamma$-coherent dictionary.
\item $\log n - \log(R^2 + \gamma)$ for a $\gamma$-Babel dictionary.
\end{itemize}
Before providing the proof of these results, it is worth noting that the conclusion and discussion conducted in the case of the gaussian kernel are still satisfied in the general case of any kernel type. 

\begin{proof}
The bounds for the $\delta$-approximate and $\gamma$-Babel dictionaries are straightforward from Theorem~\ref{th:approx} and the definition in \eqref{eq:babelD}. The lower bounds for $\gamma$-coherent and $\delta$-distant dictionaries are a bit trickier to prove. To show this, we use for the former the following relation
\begin{align*}
    H_2 &\geq
        -\log \Bigg(\frac{1}{n^2}\sum_{i=1}^n\kappa(\bx_{i},\bx_{i})
\\      & \qquad\qquad
        +\frac{\gamma}{n^2} \sum_{i=1}^n\mathop{\sum_{j=1}^n}_{j \neq i}
        \sqrt{\kappa(\bx_{i},\bx_{i})\,\kappa(\bx_{j},\bx_{j})}\Bigg)
        \\&\geq \log n -\log \big(R^2 + (n-1)\gamma R^2 \big),
\end{align*}
and for the latter the following relation
\begin{align*}
    H_2 &\geq
        -\log \Bigg(\frac{1}{n^2}\sum_{i=1}^n\kappa(\bx_{i},\bx_{i})
\\      & \qquad\qquad
        +\frac{1}{n^2} \sum_{i=1}^n\mathop{\sum_{j=1}^n}_{j \neq i}
        \sqrt{\big(\kappa(\bx_{i},\bx_{i}) - \delta^2 \big) \kappa(\bx_{j},\bx_{j})}\Bigg)
        \\&\geq \log n -\log \big(R^2 +  (n-1) R \sqrt{R^2 - \delta^2} \big).
\end{align*}
\end{proof}


\subsection{Generalization to R\'enyi and Tsallis entropies}\label{sec:entropy_Renyi}

So far, we have investigated the quadratic entropy and derived lower bounds for each diversity measure. It turns out that these results can be extended to the general R\'enyi entropy and Tsallis entropy, as shown next. Special cases of the former are listed in \tablename~\ref{tab:entropy}, including the Harley or maximum entropy which is associated to the cardinality of the set, the Shannon entropy which is essentially the Gibbs entropy in statistical thermodynamics, the quadratic entropy also called collision entropy, as well as the min-entropy which is the smallest measure in the family of R\'enyi entropies.

\begin{corollary}
Any lower bound $\zeta$ on the quadratic entropy provides lower bounds on the Hartley entropy $H_0$, the Shannon $H_1$, and the min-entropy $H_\infty$, with
\begin{equation*}
\zeta \leq H_1 \leq H_0 \qquad \text{and} \qquad  \tfrac12 \zeta \leq H_\infty.
\end{equation*}
\end{corollary}

\begin{proof}
The proof is due to the Jensen's inequality and the concavity of the R\'enyi entropy for nonnegative orders. First, the relation of the Shannon entropy is given by exploring the following inequality:
$$\sum_{j=1}^n P(\bx_{j}) \log P(\bx_{j}) \leq \log \sum_{j=1}^n \big(P(\bx_{j})\big)^2.$$
The connection to the Hartley entropy is straightforward, with $H_0 = \log n$. Finally, it is more trickier to study the min-entropy, since it is the smallest entropy measure in the family of R\'enyi entropies, as a consequence it is the strongest way to measure the information content. To provide a lower bound on the min-entropy, we use the relations
\begin{equation*}
\log \sum_{j=1}^n \big(P(\bx_{j})\big)^2 
\geq \log \max_{j=1\cdots n} \big(P(\bx_{j})\big)^2 
= 2  \log \max_{j=1\cdots n} P(\bx_{j}),
\end{equation*}
which yields the following inequality: $H_2 \leq 2 H_\infty$.
\end{proof}

Furthermore, one can easily extend these results to the class of the Tsallis entropy, also called nonadditive entropy, defined by the following expression for a given parameter $q$ (called entropic-index) \cite{Tsallis09,Tsallis_entropy}:
\begin{equation*}
	\frac{1}{q-1} \Big( 1 - \sum_{j=1}^n \big(P(\bx_{j})\big)^q \Big).
\end{equation*}
To this end, the aforementioned lower bounds on the R\'enyi entropy can be extended to the Tsallis entropy by using for instance the well-known relation $\log u \leq u-1$ for any $u \geq 0$.

As a consequence, the lower bounds on the quadratic entropy given in Sections~\ref{sec:entropy_gaussian} and \ref{sec:entropy_quadratic} can be explored to other orders of R\'enyi entropy and Tsallis entropy. 


\section{Entropy in the feature space}\label{sec:entropy_H}

By analogy with the entropy analysis in the input space conducted in Section \ref{sec:entropy}, we propose to revisit it in the feature space, as given in this section. By examining the pairwise distance between any two atoms of the investigated dictionary, we first establish in Section~\ref{sec:entropy_H_1} a topological analysis of overcomplete dictionaries. This analysis is explored in Section~\ref{sec:entropy_H_2} with the study of the entropy of the atoms in the feature space. By providing lower bounds in terms of the diversity measures, these results provide connections to the entropy analysis conducted in the previous section.

\subsection{Fundamental analysis}\label{sec:entropy_H_1}

The following theorem is used in the following section for the analysis of the atoms of a kernel dictionary.

\begin{theorem}\label{th:topo_H}
For any dictionary with a non-zero approximation measure, or a non-unit coherence measure, or a Babel measure below $r^2$, we have a low-bounded distance measure.
\end{theorem}

\begin{proof}
The proof is straightforward for a $\delta$-approximate dictionary, since
\begin{align*}
\!\!    \normH{\kappa(\bx_{i},\cdot) - \kappa(\bx_{j},\cdot)}
     & \geq 
      \!  \min_{\xi_1\cdots\xi_n} \normHBig{\kappa(\bx_{i},\cdot) -\! \mathop{\sum_{j=1}^n}_{j \neq i} \xi_j \, \kappa(\bx_{j},\cdot)}
   \\& \geq  \delta.
\end{align*}

For the coherence measure, we consider the pairwise distance in terms of kernels as given in \eqref{eq:dist_equality}. Since a $\gamma$-coherent dictionary satisfies 
\begin{equation*}
	\mathop{\max_{i,j=1\cdots n}}_{i \neq j} 
	\frac{|{\kappa(\bx_{i},\bx_{j})}|} {\sqrt{\kappa(\bx_{i},\bx_{i}) \, \kappa(\bx_{j},\bx_{j})}}
	\leq \gamma,
\end{equation*}
then we have
\begin{equation*}
	\mathop{\max_{i,j=1\cdots n}}_{i \neq j} 
	| \kappa(\bx_{i},\bx_{j}) |
	\leq \gamma 
	\mathop{\max_{i,j=1\cdots n}}_{i \neq j}
	\sqrt{\kappa(\bx_{i},\bx_{i}) \, \kappa(\bx_{j},\bx_{j})}.
\end{equation*}
Thus, $\normH{\kappa(\bx_{i},\cdot) - \kappa(\bx_{j},\cdot)}^2$ from the right-hand-side of equation \eqref{eq:dist_equality} is lower-bounded by
\begin{equation*}
	\kappa(\bx_{i},\bx_{i}) 
	- 2 \, \gamma \sqrt{\kappa(\bx_{i},\bx_{i}) \, \kappa(\bx_{j},\bx_{j})}
	 + \kappa(\bx_{j},\bx_{j}). 
\end{equation*}
Therefore, to complete the proof, it is sufficient to show that this expression is always strictly positive. Indeed, it is a quadratic polynomial of the form $u^2 -2 \gamma u v + v^2$ where $u=\sqrt{\kappa(\bx_{i},\bx_{i})}$ and $v=\sqrt{\kappa(\bx_{j},\bx_{j})}$ (this form is valid since $\kappa(\bx,\bx) = \normH{\kappa(\bx,\cdot)}^2 >0$ for any $\bx \in \X$). Considering the roots of this quadratic polynomial with respect to $u$, its discriminant is $4 \, \kappa(\bx_{j},\bx_{j}) (\gamma^2 - 1)$, which is strictly negative since $\gamma \in [\, 0 \;; 1\,[$ and $\kappa(\bx_{j},\bx_{j})$ cannot be zero. Therefore, the polynomial has no real roots, and it is strictly positive.

Finally, for any $\gamma$-Babel dictionary, we have
\begin{align*}
    \mathop{\min_{i,j=1\cdots n}}_{i \neq j} & \normH{\kappa(\bx_{i},\cdot) - \kappa(\bx_{j},\cdot)}^2
\\   
 &= \mathop{\min_{i,j=1\cdots n}}_{i \neq j} \kappa(\bx_{i},\bx_{i}) - 2 \kappa(\bx_{i},\bx_{j}) + \kappa(\bx_{j},\bx_{j})
\\   &\geq 2r^2 - 2 \, \mathop{\max_{i,j=1\cdots n}}_{i \neq j} |\kappa(\bx_{i},\bx_{j})|,
\\   &\geq 2r^2 - 2 \, \gamma,
\end{align*}
which is strictly positive when $\gamma < r^2$.
\end{proof}

\subsection{Entropy in the RKHS}\label{sec:entropy_H_2}

The entropy in the feature space provides a measure of diversity of the atoms distribution. In the following, we show that the entropy estimated in the feature space is lower-bounded, with a bound expressed in terms of a diversity measure. 

We denote by $P_{_\H}(\bx)$ the distribution associated to the kernel functions in the feature space, namely by definition $P_{_\H}(\bx) = P(\kappa(\bx,\cdot))$. The entropy in the RKHS is given by expression \eqref{eq:entropy_int} where $P(\bx)$ is substituted with $P_{_\H}(\bx)$, yielding\footnote{The expectation in a RKHS, as in \eqref{eq:entropy_int_H}, was previously investigated in the literature. The notion of embedding a Borel probability measure $P$, defined on the topological space $\X$, into a RKHS $\H$ was studied in detail in \cite{SriperumbudurVangeepuram:2010}, with $\int_\X \kappa(\bx,\cdot) \, dP(\bx)$. For an algorithmic use, see \cite{12.isit} for a one-class classifier.}
\begin{equation}\label{eq:entropy_int_H}
	\frac{1}{1-\alpha} \log \int_\X \big(P_{_\H}(\bx)\big)^\alpha \, d\bx.
\end{equation}
By approximating the integral in this expression with the set $\{\bx_{1}, \bx_{2}, \ldots, \bx_{n}\}$, we get
\begin{equation*}
	\frac{1}{1-\alpha} \log \sum_{j=1}^n \big(P_{_\H}(\bx_{j})\big)^\alpha.
\end{equation*}
The distribution $P_{_\H}(\cdot)$ is estimated with the Parzen window estimator. The use of a radial function $w(\cdot)$ defined in the feature space $\H$ yields
\begin{equation*}
	\widehat{P}_{_\H}(\bx)=\frac{1}{n}\sum_{j=1}^n w(\normH{\kappa(\bx,\cdot) - \kappa(\bx_{j},\cdot)}).
\end{equation*}

Examples of radial functions are --- up to a scaling factor to ensure the integration to one --- the gaussian, the radial-based exponential and the inverse mutliquadratic kernels, given in \tablename~\ref{tab:kernels} and applied here in the feature space. Radial kernels are monotonically decreasing in the distance, namely $\kappa(\bx_{i},\bx_{j})$ grows when $\|\bx_{i}-\bx_{j}\|$ is decreasing. This statement results from the following lemma; See also \cite[Proposition~5]{Cucker02onthe}.
\begin{lemma}\label{th:monotonic}
Any kernel $\kappa$, of the form $\kappa(\bx_{i},\bx_{j}) = g(\|\bx_{i} - \bx_{j}\|^2)$ with $g \colon (0,\infty) \to \R$, is positive definite if $g(\cdot)$ is completely monotonic, namely its $k$-th derivative $g^{(k)}$ satisfies $(-1)^k g^{(k)}(r) \geq 0$ for any $r,k\geq 0$.
\end{lemma}

\begin{theorem}\label{th:proba_estim_H}
Consider an overcomplete kernel dictionary with a lower bound $\epsilon$ on its distance measure, or any bounded diversity measure as given in Theorem~\ref{th:topo_H}. A Parzen window estimator, estimated over the dictionary atoms in the feature space, is upper-bounded by $w(\epsilon)$, where $w(\cdot)$ is the used window function.
\end{theorem}

\begin{proof}
The proof is follows from 
\begin{align*}
	\widehat{P}_{_\H}(\bx) 
	&= \frac{1}{n}\sum_{j=1}^n w( \normH{\kappa(\bx,\cdot) - \kappa(\bx_{j},\cdot)}) 
	\\&< \frac{1}{n}\sum_{j=1}^n w(\epsilon) 
	\\&= w(\epsilon),
\end{align*}
where the inequality is due to the monotonically decreasing property of the window function $w$ and Theorem~\ref{th:topo_H}.\\
\end{proof}
This theorem is the main building block of the following corollary that provides lower bounds on the entropy, with the Shannon entropy and generalizing to the R\'enyi entropy for any order $\alpha > 1$.

\begin{corollary}\label{th:entropy_H}
Consider an overcomplete kernel dictionary with a lower bound $\epsilon$ on its distance measure, or any bounded diversity measure as given in Theorem~\ref{th:topo_H}. The Shannon entropy and the generalized R\'enyi entropy for any order $\alpha > 1$ are lower bounded by $-n \, w(\epsilon) \log w(\epsilon)$ and $\frac{1}{1-\alpha} \log \big( n \, \big( w(\epsilon) \big)^\alpha \big)$, respectively, where $w(\cdot)$ is the used window function.
\end{corollary}

\begin{proof}
From Theorem~\ref{th:proba_estim_H}, we have $\widehat{P}_{_\H}(\bx) < w(\epsilon)$ for any window function $w(\cdot)$. This yields for the Shannon entropy: 
\begin{equation*}
- \sum_{j=1}^n \widehat{P}_{_\H}(\bx_{j}) \log \widehat{P}_{_\H}(\bx_{j}) 
  > -n \, w(\epsilon) \log w(\epsilon).
\end{equation*}
More generally, the R\'enyi entropy for any order $\alpha$ is estimated by
\begin{equation*}
\frac{1}{1-\alpha} \log \sum_{j=1}^n \big(\widehat{P}_{_\H}(\bx_{j})\big)^\alpha 
 > \frac{1}{1-\alpha} \log \Big( n \, \big( w(\epsilon) \big)^\alpha \Big),
\end{equation*}
where we have used Theorem~\ref{th:proba_estim_H} and $\alpha >1$.

\end{proof}

These results illustrate how the atoms of an overcomplete dictionary are uniformly spread in the feature space.


\section{Final remarks}\label{sec:final_remarks}

This paper provided a framework to examine linear and kernel dictionaries with the notion of entropy from information theory. By examining different  diversity measures, we showed that overcomplete dictionaries have lower bounds on the entropy. While various definitions were explored here, these results open the door to bridging the gap between information theory and diversity measures for the analysis and synthesis of overcomplete dictionaries, in both input and feature spaces. As of futur works, we are studying connections to the entropy component analysis \cite{KECA}, in order to provide a thorough examination and develop an online learning approach.

The conducted analysis, illustrated here within the framework of kernel-based learning algorithms, can be easily extended to other machines such as gaussian processes and neural networks. It is worth noting that this work does not devise any particular diversity measure for quantifying overcomplete dictionaries, in the same spirit as our recent work \cite{14.sparse.errors,14.sparse.eigen}.

%




\bibliography{biblio_ph,bibdesk_Paul}
\bibliographystyle{ieeetr}

\begin{biography}[{}]
{Paul Honeine} (M'07) was born in Beirut, Lebanon, on October 2, 1977. He received the Dipl.-Ing. degree in mechanical engineering in 2002 and the M.Sc. degree in industrial control in 2003, both from the Faculty of Engineering, the Lebanese University, Lebanon. In 2007, he received the Ph.D. degree in Systems Optimisation and Security from the University of Technology of Troyes, France, and was a Postdoctoral Research associate with the Systems Modeling and Dependability Laboratory, from 2007 to 2008. Since September 2008, he has been an assistant Professor at the University of Technology of Troyes, France. His research interests include nonstationary signal analysis and classification, nonlinear and statistical signal processing, sparse representations, machine learning. Of particular interest are applications to (wireless) sensor networks, biomedical signal processing, hyperspectral imagery and nonlinear adaptive system identification. He is the co-author (with C. Richard) of the 2009 Best Paper Award at the IEEE Workshop on Machine Learning for Signal Processing. Over the past 5 years, he has published more than 100 peer-reviewed papers. 
\end{biography}


\end{document}